\documentclass[graybox]{svmult}

\usepackage{type1cm}        
\usepackage{makeidx}         
\usepackage{graphicx}        
\usepackage{multicol}        
\usepackage[bottom]{footmisc}

\usepackage{bbm}
\usepackage{newtxtext}       %
\usepackage[varvw]{newtxmath}       

\usepackage{url}
\usepackage{float}

\usepackage[dvips]{epsfig}
\usepackage{subfigure}

\usepackage{algorithmicx}
\usepackage{algorithm}
\usepackage[noend]{algpseudocode}

\usepackage{ulem}

\newcommand{\bbN}{{\mathbb{N}}}
\newcommand{\bbE}{{\mathbb{E}}}
\newcommand{\bbP}{{\mathbb{P}}}
\newcommand{\bbR}{{\mathbb{R}}}
\DeclareMathOperator{\var}{Var}
\newcommand{\setU}{{\mathfrak{U}}}
\newcommand{\bsY}{{\boldsymbol{Y}}}
\newcommand{\calF}{{\mathcal{F}}}
\newcommand{\calS}{{\mathcal{S}}}
\newcommand{\calO}{{\mathcal{O}}}
\DeclareMathOperator{\cov}{Cov}
\providecommand{\argmin}{\operatorname*{argmin}}

\makeindex             

\begin{document}

\title*{Sequential Estimation using Hierarchically Stratified Domains with Latin Hypercube Sampling}

\titlerunning{SS-LHS-gPC} 
\author{Sebastian Krumscheid and Per Pettersson}

\institute{Sebastian Krumscheid \at Karlsruhe Institute of Technology, 76131 Karlsruhe, Germany \email{krumscheid@kit.edu}
\and Per Pettersson \at NORCE Norwegian Research Centre, N-5838 Bergen, Norway \email{pepe@norceresearch.no}}
%
%
\maketitle

\abstract{%
Quantifying the effect of uncertainties in computationally complex systems where only point evaluations in the stochastic domain but no regularity conditions are available is limited to sampling-based techniques. This work presents an adaptive sequential stratification estimation method that uses Latin Hypercube Sampling within each stratum. The adaptation is achieved through a sequential hierarchical refinement of the stratification, guided by previous estimators using local (i.e., stratum-dependent) variability indicators based on generalized polynomial chaos expansions and Sobol decompositions. For a given total number of samples $N$, the corresponding hierarchically constructed sequence of Stratified Sampling estimators combined with Latin Hypercube sampling is adequately averaged to provide a final estimator with reduced variance. Numerical experiments illustrate the procedure's efficiency, indicating that it can offer a variance decay proportional to $N^{-2}$ in some cases.}

\section{Introduction}
\label{KrPe_sec:intro}

In many applications, including problems in flood modelling~\cite{Hajihassanpour_etal_23} and large-scale CO$_2$ storage simulations~\cite{Gasda_etal_22}, models that describe complex phenomena are subject to uncertainties. Often, these models have a black-box character, in the sense that the model can only be evaluated pointwise in the stochastic domain without any other information regarding the mapping of the stochastic input to a model output. A reliable assessment of the effects of uncertainties on system outputs is then only possible via sampling-based methods. In some cases, e.g., in multi-physics problems or when modeling fractured porous media, the model structure, and its complexity prohibit using a hierarchy of different degrees of fidelity. In these cases, efficient estimators that use different fidelity models to systematically reduce the variance of estimators of statistics of a random system's output, such as multi-fidelity~\cite{Peherstorfer_etal_18} and multi-level techniques~\cite{Giles_15}, cannot be used. Other sampling-based variance reduction techniques, which are applicable in these situations, are methods that use a partition of the stochastic domain. Procedures such as Stratified Sampling~\cite{Asmussen_Glynn_07,Kroese_etal_11}, which are typically restricted to low dimensions, and the related Latin Hypercube Sampling for higher dimensions~\cite{McKay_etal_79,Stein_87}, have the potential to reduce the variance compared to a standard Monte Carlo (SMC) estimator when using the same number of samples. Indeed, these techniques can be very effective when the partition of the stochastic domain is aligned with the random system's high variability region, which has led to recent developments on adaptive procedures that guide the partition of the stochastic domain; see, e.g., \cite{Etore_etal_11, Shields_18, KrPe_ref:PeKr_22, Song_Kawai_22}.

Aside from the challenges mentioned above, evaluating a model can often be computationally expensive. In these cases, the cost of repeatedly evaluating the model is significantly more expensive than the costs of post-processing obtained samples to guide adaptive sampling strategies. Here, we introduce an adaptive sampling-based estimator that uses a combination of Stratified Sampling and Latin Hypercube Sampling (LHS) specifically designed for such problems where generating samples is computationally expensive.
In contrast to previous works where the adaptation is achieved through a sequential sampling approach, a sequential approach is not easily achieved for LHS, as one would lose favorable LHS properties after adding samples. We remedy this by introducing a sequential estimation approach, where we construct a sequence of Stratified Sampling estimators, each using its own LHS. 
For example, let $N = \sum_{\ell=1}^L N^{(\ell)}$ be the total number of samples aggregated from $L$ batches of sizes $N^{(\ell)}$, and denote by $\hat{\mu}^{\text{Strat}}_{\text{LHS},\ell}$ the Stratified Sampling estimator using an LHS sample of size $N^{(\ell)}$.
The final estimator is then obtained through an optimally weighted average of the sequence of estimators $\hat{\mu}^{\text{Strat}}_{\text{LHS},\ell}$, $\ell = 1,\dots, L$. In this approach, adaptivity for the partition of the stochastic domain is realized by a hierarchical refinement strategy based on a local sensitivity analysis of the current estimator in the sequence before computing the next one. 
Indeed, given an estimator $\hat{\mu}^{\text{Strat}}_{\text{LHS},\ell}$, the choice of stratification used for the next estimator $\hat{\mu}^{\text{Strat}}_{\text{LHS},\ell+1}$ can be informed by, e.g., local variance contributions from the current estimator $\hat{\mu}^{\text{Strat}}_{\text{LHS},\ell}$.
This can, for example, be done via a Sobol decomposition~\cite{Sobol_93}, which, however, is computationally expensive if performed via numerous distinct Monte Carlo integrals. By computing a generalized polynomial chaos expansion of the function of interest, the Sobol sensitivity indices can efficiently be obtained via a fast post-processing procedure~\cite{Sudret_08}, leading to a computationally efficient adaptive estimation procedure.
An implementation of the proposed methodology is available at \url{https://github.com/massperp/Sequential-Stratified-LHS}.

An adaptive LHS procedure for  sequentially adding new samples to a domain while preserving LHS properties was proposed in~\cite{Tong_06}. Another related work is~\cite{Shields_16}, where a sampling approach is introduced that adds new samples sequentially to jointly LHS and Stratified Sampling designs. These approaches are distinctly different from the one proposed in this paper, where we consider a sequence of estimators rather than sequential sampling within a single estimator, and multi-sample LHS within a sequentially refined stratification design instead of a joint LHS and Stratified Sampling domain.

The rest of the paper is structured as follows. Sect.~\ref{sec:basics} describes the required basics of both Stratified Sampling and Latin Hypercube Sampling. Sect.~\ref{sec:adaptive:proc} then describes the novel hierarchical procedure, detailing the sequential estimation approach, and the identification of local variability indicators that guide the hierarchical adaptation. Numerical experiments are reported in Sect.~\ref{sec:numerics}. These examples indicate, in particular, that the novel method is superior to both Monte Carlo and Latin Hypercube Sampling and can even achieve a variance of order $N^{-2}$ in some cases, where $N$ is the sample size. Finally, a conclusion is offered in Sect.~\ref{sec:conclusion}.


\section{Combining Stratified Sampling and Latin Hypercube Sampling}
\label{sec:basics}
We consider computing $\mu = \bbE{\bigl(f(\bsY)\bigr)}$ for a uniformly distributed $d$-dimensional random vector
$\bsY = (Y_1,\dots, Y_d)\sim \text{Uni}(\setU)$, $\setU = [0, 1]^d$,  on some probability space $(\Omega,\calF,\bbP)$ and  a given scalar-valued measurable function $f\colon\bbR^d\to\bbR$, which is such that $f(\bsY)$ has finite variance. Throughout this work, we assume that $f$ is a complex computational model that may be expensive to evaluate and which we can only access through point evaluations. In particular, we do not assume that derivative information on $f$ is available. Consequently, the goal is to efficiently approximate $\mu$ using a sampling-based procedure. Here, we will combine Stratified Sampling and Latin Hypercube sampling in an adaptive procedure. Before we describe the adaptive sampling principle, we first present the basic building blocks required to blend Stratified Sampling and Latin Hypercube Sampling.

\subsection{Stratified Sampling}
\label{sec:strat_samp}
Stratified Sampling is a sampling-based approach for approximating $\mu = \bbE{\bigl(f(\bsY)\bigr)}$ with reduced variance compared to Monte
Carlo sampling \cite{Asmussen_Glynn_07,Kroese_etal_11,KrPe_ref:PeKr_22}. This is achieved by decomposing the stochastic domain $\setU$
into multiple disjoint regions, so-called strata. Let $\calS$ be a stratification of the domain $\setU$, i.e., $\mathfrak{U} = \cup_{S\in\calS}S $ with $S\cap T = \emptyset$ for $S,T\in\calS$ with $S\not=T$. Denoting  by $p_S = \bbP(\bsY\in S)$ the measure (or ``size'') of stratum $S$, by the law of total probability we can write
\begin{equation}\label{eqn:stratification:start}
 \mu = \bbE{\bigl(f(\bsY)\bigr)} =  \sum_{S\in\calS} p_S \bbE{\bigl(f(\bsY)| \bsY\in S\bigr)} = \sum_{S\in\calS} p_S \bbE{(f_S)}\;,
\end{equation}
where $f_S = f(\bsY)| \bsY\in S$ denotes the random variable that has the distribution of $f(\bsY)$ conditioned upon $\bsY\in S$. The common approach to Stratified Sampling is to approximate the expectation $\bbE{(f_S)}$ by SMC with $N_S$ samples in each stratum $S$, assuming $p_S$ is known, leading to
\begin{equation}\label{eqn:strat_samp_est}
\hat{\mu}^{\text{Strat}}_{\text{MC}} = \sum_{S \in \calS} \frac{p_{S}}{N_{S}}  \sum_{j=1}^{N_S} f_S^{(j)}\;,
\end{equation}
where $f_S^{(j)}$, $j=1,\dots, N_S$, are independent and identically distributed (i.i.d.) samples from the conditional random variable $f_S$. Consequently, $\hat{\mu}^{\text{Strat}}_{\text{MC}}$ is an unbiased estimator of $\mu$ and its variance is 
\begin{equation} \label{eq:ss_estimator:var:general}
  \var{(\hat{\mu}^{\text{Strat}}_{\text{MC}})} = \sum_{S\in\calS}\frac{p_S^2 \sigma_S^2}{N_S}\;,
\end{equation}
with conditional, i.e., local variance $\sigma_S^2 := \var{(f_S)} = \var{\bigl(f(\bsY)|\bsY\in S\bigr)}$. 

In practice, a rule for the number of samples $N_S$ in each stratum $S\in\calS$ needs to be selected. Let $N = \sum_{S\in\calS} N_S$ denote the total number of samples used. The common choices for distributing these $N$ samples across the strata in $\calS$ are proportional allocation and optimal sample allocation \cite{Asmussen_Glynn_07}, where the number of samples in a stratum $S$ is chosen proportional to its size $p_S$ 
or, respectively, such that \eqref{eq:ss_estimator:var:general} is minimized subject to the condition $N = \sum_{S\in\calS} N_S$, leading to the corresponding number of samples
\begin{equation*}
    N^{\text{prop}}_S = p_S N\quad\text{and}\quad N^{\text{opt}}_S = \frac{p_S \sigma_S}{\sum_{T\in \calS} p_T \sigma_T} N\;,
\end{equation*}
respectively. Recently, hybrid allocation rules have been introduced that bridge the gap between proportional and optimal allocation \cite{KrPe_ref:PeKr_22}. Finally, while most of the concepts presented in this work can be generalized, we will henceforth consider only stratifications $\calS$ that consist of hyperrectangles.

\subsection{Latin Hypercube Sampling within Stratified Sampling}
\label{sec:LHS-within-SS}
Although an SMC estimator for $\bbE{(f_S)}$ in \eqref{eqn:stratification:start} is natural and  commonly used, other estimators are possible too. Indeed, using the stratification $\calS$, a general unbiased  estimator of $\mu = \bbE{\bigl(f(\bsY)\bigr)}$ can be defined through
\begin{equation}\label{eqn:strat_samp:generic}
\hat{\mu}^{\text{Strat}} = \sum_{S \in \calS} p_{S} \hat{\mu}_S\;,
\end{equation}
for any unbiased estimators $\hat{\mu}_S$ of $\bbE{(f_S)}$, $S\in\calS$. As we focus on sampling-based estimators, we consider $\hat{\mu}_S$ to be an estimator using $N_S$ samples, but remark that other choices are possible. We also suppose that the samples $f_S^{(j)}$, $j=1,\dots, N_S$, in stratum $S$ and the samples $f_T^{(k)}$, $k=1,\dots, N_T$, in stratum $T$ are independent for $S\not=T$, so that
\begin{equation} \label{eq:ss_estimator:generic:var:general}
  \var{(\hat{\mu}^{\text{Strat}})} = \sum_{S\in\calS}p_S^2 \var{(\hat{\mu}_S)}\;.
\end{equation}
In this work, we will use the generalized Stratified Sampling viewpoint~\eqref{eqn:strat_samp:generic} with an LHS estimator $\hat{\mu}_S$ in each stratum $S\in\calS$.  
LHS is somewhat related to stratified sampling \cite{Asmussen_Glynn_07} and can overcome the computationally prohibitive nature of (uniform) stratified sampling in larger dimensions $d$. The idea of LHS in stratum $S$ is to stratify each component of the conditioned vector $\bsY_S := (Y_{1,S},\dots, Y_{d,S})| \bsY \in S$ but not the  whole sampling domain (i.e., stratum) $S$. In particular, $N_S$ points $\bsY_{S,N_S}^{(j)}$, $j=1,\dots N_S$, are drawn uniformly in stratum $S$ in such a way that each component of $\bsY_S$ itself is stratified with $N_S$ uniform LHS-strata and one point per LHS-stratum. In other words, LHS stratifies each marginal distribution of $\bsY_S$. 

Using the Latin Hypercube samples $\bsY_{S,N_S}^{(j)}$, $j=1,\dots N_S$,  of $\bsY_S$, the LHS estimator of $\mu_S$ is then given by
\begin{equation}
\label{eqn:strat_samp:lhs:stratum}
\hat{\mu}_{S}^{\text{LHS}} = \frac{1}{N_S}  \sum_{j=1}^{N_S} \tilde{f}_{S, N_S}^{(j)}\;,\quad \tilde{f}_{S,N_S}^{(j)} := f(\bsY_{S,N_S}^{(j)})\;, j=1,\dots, N_S\;.
\end{equation}
While the LHS estimator $\hat{\mu}_{S}^{\text{LHS}}$ is an unbiased estimator of $\bbE(f_S)$, the LHS sample is correlated since the $k$th component of $\bsY_S^{(j)}$ cannot be in the same LHS-stratum as its  $\ell$th component for $k\not=\ell$. This correlation leads to
\begin{equation}
\label{eqn:strat_samp:lhs:stratum:var}
\var{(\hat{\mu}_{S}^{\text{LHS}})} = \frac{\sigma_S^2}{N_S}  + \frac{N_S - 1}{N_S} \cov{\big(\tilde{f}_{S,N_S}^{(1)}, \tilde{f}_{S,N_S}^{(2)}\bigr)}\;,
\end{equation}
 indicating that the LHS estimator only offers a lower variance compared to SMC for a fixed sample size  $N_S$, if $\cov{\big(\tilde{f}_{S,N_S}^{(1)}, \tilde{f}_{S,N_S}^{(2)}\bigr)}<0$. This is ensured, for example, whenever the restriction $f_S = {f|}_{S}$ of $f$ to stratum $S$ is monotonic in each of its arguments~\cite{McKay_etal_79}. However, for $N_S$ large, the variance will never increase compared to SMC. Indeed, $\lim_{N_S\to\infty} N_S \cov{\big(\tilde{f}_{S,N_S}^{(1)}, \tilde{f}_{S,N_S}^{(2)}\bigr)} \le 0$ for any $f\in L^2(S)$ \cite{Stein_87}, so that $\lim_{N_S\to\infty} N_S\var{(\hat{\mu}_{S}^{\text{LHS}})} \le \sigma_S^2$
 in view of \eqref{eqn:strat_samp:lhs:stratum:var}. Moreover, the LHS estimator is particularly efficient whenever $f$ is close to an additive function \cite{Stein_87}, in the sense that
 $N_S \var{(\hat{\mu}_{S}^{\text{LHS}})} = {\Vert r_f \Vert}^2_{L^2(S)} + o(1)$ as $N_S\to\infty$, where $r_f = \argmin_{\phi \in U}{\Vert f- \phi \Vert}^2_{L^2(S)}$ with $U = \bigl\{g\in L^2(S)\;\bigl|\bigr.\; g(y) = \sum_{k=1}^dg_k(y_k), g_k\colon\bbR\to\bbR\bigr\}$.

\section{Sequential Estimation with Latin Hypercube Sampling}
\label{sec:adaptive:proc}
Previous approaches to adaptive stratified sampling use a sequential sampling strategy~\cite{Etore_etal_11, KrPe_ref:PeKr_22, Song_Kawai_22}. For example, in \cite{KrPe_ref:PeKr_22} the total number of samples $N$ is not generated at once, but sequentially in a batch-wise manner instead. A batch of new i.i.d.\ (i.e., Monte Carlo) samples  are added to strata with already existing
samples.
Before repeating this procedure for the next batch, the variance contribution of each stratum is estimated and the stratification is refined to systematically reduce the estimator's variance. This sequential sampling approach cannot be extended to samples generated by LHS instead of random i.i.d.\ sampling. This is because LHS does not produce nested samples, in the sense that in stratum $S$, the Latin Hypercube samples $\bsY_{S,N_S}^{(j)}$, $j=1,\dots N_S$,  
will not be a valid subset of an LHS when using $M_S > N_S$ samples. To leverage the benefits of LHS within an adaptive Stratified Sampling approach, we will therefore use a sequential estimation procedure instead.

\subsection{Sequential Estimation}
The idea of this work to construct a generalized stratified sampling estimator using LHS that is adapted to the problem at hand, i.e., tailored to the approximation of $\mu = \bbE{\bigl(f(\bsY)\bigr)}$, is to also use a batch-wise strategy instead of using the total number of samples $N$ directly. Because of the non-nested nature of LHS, each batch will be used to construct a separate Stratified Sampling estimator combined with LHS. That is, let $N = \sum_{\ell=1}^L N^{(\ell)}$ be a decomposition of the total number of samples with batches of sizes $N^{(\ell)}\in\bbN$ and denote by $\hat{\mu}^{\text{Strat}}_{\text{LHS}, \ell}$, $\ell=1,\dots, L$, the Stratified Sampling estimator using $N^{(\ell)}$ LHS samples. Each estimator $\hat{\mu}^{\text{Strat}}_{\text{LHS}, \ell}$ is called stratified LHS estimator (S-LHS estimator for short). The motivation for creating these batch-wise estimators sequentially is that for a given  $\hat{\mu}^{\text{Strat}}_{\text{LHS}, \ell}$, the stratification for the next estimator $\hat{\mu}^{\text{Strat}}_{\text{LHS}, \ell+1}$ can be adapted, e.g., in a variance minimizing manner (to be made precise in what follows), using information about the current estimator $\hat{\mu}^{\text{Strat}}_{\text{LHS}, \ell}$. This refinement procedure is expected to yield increasingly better adapted generalized Stratified Sampling estimators. However, the batch-wise construction means that the final S-LHS estimator $\hat{\mu}^{\text{Strat}}_{\text{LHS}, L}$ will have a total number of samples $N^{(L)}$ that is much smaller than $N$. Recall that each sample involves an evaluation of the possibly expensive model $f$. To not ``waste'' samples that were used to construct the earlier S-LHS estimators in this sequential procedure, which may already contain valuable insight about $\mu = \bbE{\bigl(f(\bsY)\bigr)}$, we will combine the entire ensemble of S-LHS estimators into a global one. A weighted average of all estimators is obtained as summarized in the following result, whose proof is constructive and hence added for completeness

\begin{lemma}
Let $\{\hat{\mu}_1,\dots \hat{\mu}_L\}$, $L\in\bbN$, be an ensemble of mutually independent, unbiased estimators of $\mu$, that is $\bbE(\hat{\mu}_\ell) = \mu$ for $1\le \ell\le L$, each with finite variances $v_\ell = \var(\hat{\mu}_\ell)\ge 0$. Consider the weighted ensemble estimator 
\begin{equation}
\label{eq:weighted_estimator}
    E_L(\alpha) = \sum_{\ell = 1}^L \alpha_\ell \hat{\mu}_\ell,
\end{equation}
for weights $\alpha\equiv (\alpha_1,\dots, \alpha_L)\in\bbR^L$ with $\alpha_\ell \ge 0 $ for $1\le \ell\le L$ and such that $\sum_{\ell=1}^L\alpha_\ell = 1$. Then $E_L(\alpha)$ is an unbiased estimator of $\mu$ for any weight vector $\alpha$. Moreover, there exists a weight vector $\alpha^\ast\in\bbR^L$ to be identified in the proof, that minimizes $\alpha \mapsto \var\bigl(E_L(\alpha)\bigr)$.
\end{lemma}
\begin{proof}
The fact that $E_L(\alpha)$ is an unbiased estimator of $\mu$ follows immediately from the unbiasedness of each $\hat{\mu}_\ell$ together with $\sum_{\ell=1}^L\alpha_\ell = 1$ for any weight vector $\alpha\equiv (\alpha_1,\dots, \alpha_L)$. Next, let $v_\ell := \var(\hat{\mu}_\ell)$. We then distinguish two cases. First, if there exists $k\in\{1,\dots, L\}$ such that $v_k = 0$, then $\alpha^\ast_\ell = \delta_{k,\ell}$, where $\delta_{k,\ell}$ is the Kronecker delta, leads to a weighted ensemble estimator with zero variance. 

Second, suppose that $\min_{1\le \ell\le L} v_l >0$. Without loss of generality, we also assume that $L\ge 2$. Using the condition $\sum_{\ell=1}^L\alpha_\ell = 1$, we can write the weighted ensemble estimator's variance as
\begin{equation*}
    \var\bigl(E_L(\alpha)\bigr) = \sum_{\ell = 1}^L {\alpha_\ell}^2 v_\ell = \sum_{\ell = 1}^{L-1} {\alpha_\ell}^2 v_\ell + {\left(1-\sum_{\ell=1}^{L-1}\alpha_\ell\right)}^2v_L =: V_L(\tilde\alpha)\;,
\end{equation*}
where $\tilde\alpha = (\tilde\alpha_1,\dots, \tilde\alpha_{L-1})\in\bbR^{L-1}$ with $\tilde\alpha_\ell = \alpha_\ell$ for $1\le \ell \le L-1$. Introducing the notation $M := \operatorname{diag}(v_1,\dots, v_{L-1}) + v_L \boldsymbol{1}\boldsymbol{1}^T$ and $b := v_L \boldsymbol{1}$ with  $\boldsymbol{1} := (1,\dots, 1)^T\in\bbR^{L-1}$, we can write the function $V_L\colon\bbR^{L-1}\to\bbR$ as $V_L(\tilde\alpha) = \tilde\alpha^T M \tilde\alpha - 2 \tilde\alpha^T b + v_L$. The matrix $M$ is symmetric. Moreover, $v_L\boldsymbol{1}\boldsymbol{1}^T$ is a rank-one matrix, whose only non-zero eigenvalue is equal to $\operatorname{tr}(v_L\boldsymbol{1}\boldsymbol{1}^T) = v_L (L-1) > 0$ due to the hypothesis $\min_{1\le \ell\le L} v_l >0$. That is, $v_L\boldsymbol{1}\boldsymbol{1}^T$ is positive semi-definite. Since the matrix $\operatorname{diag}(v_1,\dots, v_{L-1})$ is positive definite in view of the hypothesis, it follows that the matrix $M$ is positive definite. Consequently, the quadratic function $V_L$ is strictly convex and has a unique global minimum attained at $\tilde\alpha^\ast$, which satisfies the necessary condition $M\tilde\alpha^\ast = b$. 

It remains to show that $\tilde\alpha^\ast$ leads to an admissible weight vector, in the sense that $\tilde\alpha_\ell^\ast\ge 0$ for $\ell\in\{1,\dots, L-1\}$, and $\sum_{\ell=1}^{L-1}\tilde\alpha_\ell^\ast \le 1$. Indeed, from the Sherman--Morrison formula \cite{KrPe_ref:Ba_51}, we find
\begin{equation}
\label{KrPe_eq:optimal:weights:avg}
    \tilde\alpha^\ast = M^{-1}b = \left( D^{-1} - \frac{v_L D^{-1}\boldsymbol{1}\boldsymbol{1}^T D^{-1}}{1 + v_L \boldsymbol{1}^TD^{-1}\boldsymbol{1}} \right)b = \frac{1}{\sum_{\ell=1}^{L}\frac{1}{v_\ell}}
    \begin{pmatrix}
    \frac{1}{v_1}\\\vdots\\\frac{1}{v_{L-1}}
    \end{pmatrix}\;,
\end{equation}
where $D^{-1}=\operatorname{diag}(v_1^{-1},\dots, v_{L-1}^{-1})$, 
from which both conditions immediately follow. Consequently, the optimal weight vector $\alpha^\ast\in\bbR^L$ is given by
\begin{equation}
\label{KrPe_eq:optimal:weights:avg:final}
    \alpha^\ast_k := \begin{cases} 1- \sum_{\ell=1}^{L-1}\tilde\alpha^\ast_\ell  \;,&k = L\;,\\ \tilde\alpha^\ast_k\;,&\text{else},
    \end{cases}
\end{equation}
for $k\in\{1,\dots, L\}$. Finally, the optimally weighted ensemble estimator’s variance is
\begin{equation*}
\begin{aligned}
\var\bigl(E_L(\alpha^\ast)\bigr) &= V_L(\tilde\alpha^\ast) = {(\tilde\alpha^\ast)}^T M \tilde\alpha^\ast - 2 {(\tilde\alpha^\ast)}^T b + v_L\\
&= v_L - {(\tilde\alpha^\ast)}^T b = v_L\left( 1 - {(\tilde\alpha^\ast)}^T\boldsymbol{1}\right) = \frac{1}{\sum_{\ell=1}^{L}\frac{1}{v_\ell}} \;,
\end{aligned}
\end{equation*}
which concludes the proof.
\end{proof}
To illustrate the preceding lemma, consider the weighted average of $L=2$ estimators $\hat\mu_1$ and $\hat\mu_2$, each with a non-zero variance $v_1$ and $v_2$, respectively. Using formula \eqref{KrPe_eq:optimal:weights:avg}, we find that the optimally weighted ensemble estimator is
\begin{equation*}
    \hat{E}_2 := \frac{v_2}{v_1 + v_2} \hat\mu_1 + \frac{v_1}{v_1 + v_2} \hat\mu_2\;.
\end{equation*}
For the special case of $\hat\mu_1$ and $\hat\mu_2$ being standard Monte Carlo estimators of $\mu=\bbE(Q)$ with $M_1$ and $M_2$ i.i.d.\ samples, respectively, the optimally weighted ensemble estimator coincides, unsurprisingly, with a plain Monte Carlo estimator using $M_1 + M_2$ i.i.d.\  samples, since $\frac{v_1}{v_1 + v_2} = \frac{M_2}{M_1 + M_2}$ and $\frac{v_2}{v_1 + v_2} = \frac{M_1}{M_1 + M_2}$. 

Next, we will introduce concepts required to guide the construction of the sequence of S-LHS estimators using hierarchical stratification refinements.

\subsection{Sensitivity Analysis using Local Surrogate Methods}

To construct a sequence of successively refined S-LHS estimators, we introduce a framework for informing the stratification of a member of the sequence of estimators, based on its immediate predecessor. 
This is achieved by computing local sensitivities of all strata with respect to the random input variables, which allows choosing the best new stratification among a set of candidate stratifications. As the function $f$ is only known via its discrete sampling points, a surrogate function based on generalized Polynomial Chaos expansions will be briefly described first. This is followed by an outline of a sensitivity analysis using a Sobol decomposition. Moreover, we establish a connection between generalized Polynomial Chaos coefficients and Sobol sensitivities. Finally, the effective dimension of the function $f$ will be discussed, which is closely related to the Sobol decomposition and indicative to explain the performance of LHS in general, and applicable to the numerical test cases in Sect.~\ref{sec:numerics} in particular.

\subsubsection{Surrogate methods using generalized Polynomial Chaos expansion}
\label{sec:surr_gPC}
The generalized Polynomial Chaos (gPC) expansion in orthogonal polynomials (c.f.~\cite{Xiu_Karniadakis_02}) provides a useful representation of functions with finite second-order moments. For example, the gPC expansion of the function $f$ restricted to a given stratum $S\in\calS$ can be written
\begin{equation}
\label{eq:gPC_expansion}
    f^{\text{gPC},S}(\bsY_{S}) = \sum_{ \mathbf{m} \in \mathbb{Z}_{0+}^{d} } f^{\text{gPC},S}_{\mathbf{m}} \psi^{S}_{\mathbf{m}} (\bsY_{S})\;,
\end{equation}
where $\mathbb{Z}_{0+}^{d}$ denotes the index set of $d$-tuples of non-negative integers, $f^{\text{gPC},S}_{\mathbf{m}}$ are the gPC coefficients to be determined, and the orthogonal polynomials $\psi^{S}_{\mathbf{m}} (\bsY_{S})$ will be described next. 
Assuming that the strata defined in Sect.~\ref{sec:strat_samp} are hyperrectangles (or tensor product spaces in general), an orthonormal polynomial basis $\{ \psi^S_{k,m}(Y_{k,S})\}_{m=0}^{\infty}$ in $Y_k$ ($k=1,\dots, d$) restricted to stratum $S$ can be introduced following the ideas of multi-element gPC~\cite{Wan_Karniadakis_05, Wan_Karniadakis_06}. A multidimensional orthonormal basis can be formed by products of the one-dimensional basis functions, i.e.,  $\psi^S_{\mathbf{m}}(\bsY_S) \equiv  \Pi_{k=1}^{d} \psi^S_{k,m_{k}}(Y_{k,S})$ with $\mathbf{m} = (m_1,m_2,\dots, m_d) \in \mathbb{Z}_{0+}^{d}$. In this work, we employ a total order basis, i.e., restrict the indices to $\left\| \mathbf{m} \right\|_{1}\leq p$ for some positive integer $p$, but other choices are possible. The resulting multivariate basis functions are orthonormal with respect to the local PDF $\rho_{S}$,
\[
\int_{S} \psi_{\mathbf{m}}^{S}(\bsY_S) \psi_{\mathbf{n}}^{S}(\bsY_S) \rho_{S}(\bsY_S)\textup{d} \bsY_S = 
\left\{ 
\begin{array}{ll}
1 & \mbox{if } \mathbf{m}=\mathbf{n}\;,\\
0 & \mbox{if } \mathbf{m} \neq \mathbf{n}\;.
\end{array}
\right.
\]
In this work, one-dimensional basis functions are computed dimension by dimension using Stieltjes procedure~\cite{Stieltjes_84}, as described in, e.g., (2.13)-(2.15) in~\cite{Wan_Karniadakis_06}, i.e., employing the recurrence relation
\begin{align}
\begin{split}
\label{eq:rec_coeff}
\psi^S_{k,m_{k}+1}(Y_{k,S}) &= (Y_{k,S}-\alpha_{m_{k}})\psi^S_{k,m_{k}}(Y_{k,S}) - \beta_{m_k} \psi^S_{k,m_{k}-1}(Y_{k,S}),\\
\psi^S_{k,0}(Y_{k,S}) &= 1, \quad \psi^S_{k,-1}(Y_{k,S}) = 0,
\end{split}
\end{align}
with the recurrence coefficients obtained from the expressions

\[
\alpha_{m_{k}} = \frac{\bbE{\bigl(  Y_{k,S} \psi^S_{k,m_{k}} \psi^S_{k,m_{k}} \bigr) }}{\bbE{ \bigl( \psi^S_{k,m_{k}} \psi^S_{k,m_{k}} \bigr) }},
\quad
\beta_{0}=1, \quad
\beta_{m_{k}} = \frac{\bbE{\bigl(  \psi^S_{k,m_{k}} \psi^S_{k,m_{k}} \bigr) }}{\bbE{ \bigl( \psi^S_{k,m_{k}-1} \psi^S_{k,m_{k}-1} \bigr) }},
\]
where the expectations can be computed exactly using Gauss-type quadrature, since all integrands are polynomials.

As the stochastic input domain is assumed to be transformed to the unit hypercube with uniform distribution, all one-dimensional basis functions will be rescaling of Legendre polynomials. For the more general case of high-dimensional non-uniform distributions, we refer to the recent development of multivariate generalizations of Stieltjes algorithm~\cite{Liu_Narayan_22}.

\subsubsection{Sobol Decomposition and Sensitivity Analysis}
Sobol decompositions provide a means to decompose a function of uniform random variables into variance contributions from all subsets of its arguments~\cite{Sobol_93}.
The Sobol decomposition of the function $f$ restricted to stratum $S$ is given by
\[
f^{\text{Sobol},S}(\bsY_S) = f_{\emptyset, S} +
\sum_{T \subset \{ 1,2,..., d\} } f_{T, S}(\bsY_{T,S}),
\]
where the terms $f_{T, S}$ are orthogonal with respect to the usual inner product on $L^2(S)$ and have zero mean (except for the constant $f_{\emptyset}$). Here, $\bsY_{T,S}$ denotes the components of the vector $\bsY$ indexed by $T$. The Sobol decomposition of  $f$ admits a corresponding variance decomposition:
\begin{equation}   
\var(f(\bsY_S)) =  \sum_{T \subset \{ 1,2,..., d\} } \sigma_{T,S}^2,
\label{eq:Sobol-var-dec}
\end{equation}
where $\sigma_{T,S}^2$ is the variance contribution from the variables in the subset $T$ of all random variables.
This Sobol variance decomposition allows direct determination of the sensitivities of all variables and all subsets of them to the total variance.
However, computing all $2^d$ distinct variance contributions by equally many direct Monte Carlo estimators is computationally expensive. Instead, using that there is a direct correspondence between subsets of gPC terms and Sobol terms, i.e., 
\begin{equation}
\label{eq:Sobol-gPC}
\sigma_{T,S}^2 = \sum_{\mathbf{m} \in \mathcal{I}_T} \left( f^{\text{gPC},S}_{\mathbf{m}} \right)^2, \mbox{ where } \mathcal{I}_{T} = \Bigl\{ \mathbf{m}=(m_1,\dots,m_d) \in \mathbb{Z}_{0+}^{n} : 
\left\{
\begin{array}{ll}
m_k = 0  & \mbox{ if } k \notin T\\
m_k > 0 & \mbox{ if }  k \in T
\end{array}
\right.
\Bigr\},
\end{equation} 
one can directly evaluate the Sobol sensitivities from the gPC coefficients~\cite{Sudret_08}. In other words, the Sobol index $S_{T}$ is the sum of squares of the presumably already computed gPC coefficients of those basis functions that are functions of the random variables defined by the set $T$.

\subsubsection{Effective Dimension}
The Sobol decomposition is closely related to the concept of effective dimension.
Following the definitions on p.~35 in~\cite{Caflisch_etal_97}, 
the effective dimension of the function $f$ restricted to stratum $S$ in the superposition sense is the smallest integer $d_{\text{sup}}$ such that
\[
\sum_{ \substack{ T \subset \{1,2,\dots, d\} \\ |T| \leq d_{\text{sup}} }} \sigma^2_{T,S}
\geq \alpha \var(f(\bsY_S))\;, 
\]
where $\sigma^2_{T,S}$ are terms in the variance decomposition~\eqref{eq:Sobol-var-dec}, and 
the threshold $\alpha$ is a constant set to some value in the unit interval depending on application. For the numerical study that follow, we will use $\alpha=0.99$, but we do not expect this value to affect numerical results much, e.g., by setting $\alpha=0.9$ or $\alpha=0.999$ instead.
Similarly, the effective dimension in the truncation sense is the smallest integer $d_{\text{tr}}$
satisfying 
\[
\sum_{T \subseteq \{1,\dots, d_{\text{tr}} \} } \sigma^2_{T,S}
\geq \alpha \var(f(\bsY_S))\;. 
\]
The notion of effective dimension is of interest in the current work for two reasons. First, a function with low $d_{\text{sup}}$ contains only low-order Sobol decomposition terms, which has an impact on the sensitivity estimates. Second, a function with low $d_{\text{sup}}$ or $d_{\text{tr}}$ are particularly suited for LHS sampling, as additive variance effects are filtered  out as described at the end of Sect.~\ref{sec:LHS-within-SS} and in~\cite{Stein_87}.
 
\subsection{A Sobol Sensitivity informed Sequence of hierarchical S-LHS Estimators} 
\label{sec:SS-LHS-Sobol}

Now we present a method for combining Stratified Sampling with LHS properties using stratum-wise sensitivity estimates via Sobol decomposition and gPC approximation to construct a sequence of estimators. To that means, we need to present an explicit stratification refinement criterion, based on a previously computed S-LHS estimator. Finally, once we have all members of the sequence of estimators, we need to estimate the variance of each estimator so that we can compute the weights according to~\eqref{KrPe_eq:optimal:weights:avg}--\eqref{KrPe_eq:optimal:weights:avg:final}, and to assemble a weighted sequential S-LHS estimator based on gPC indicators, denoted SS-LHS-gPC estimator. The same variance estimates will be expanded by means of sensitivities to be used as the criterion for refinement of the stratifications. The variance of Stratified Sampling estimators is  
dependent on the sampling density, so first we need to establish an appropriate sample allocation rule for the SS-LHS-gPC estimator.

Optimal allocation of samples requires knowledge of local variances, and if those are inaccurately approximated, the performance can be significantly inhibited~\cite{Cochran_77}. Proportional allocation is a more safe option, but we wish to use more of the information from previous estimators when allocating samples. To this end, we introduce the heuristic criterion that a stratum that is small has this property due to previous splittings, indicating that it is located in a region of stochastic space with high variability. Hence, it should have a relatively high sample allocation rate. We set this rate to be inversely proportional to the size of the stratum, and consequently allocate a constant number $\bar{N}$ of samples per stratum, independent of its size. This also implies that the total number of samples of a given estimator in the sequence is linear in the number of strata. There is otherwise significant flexibility in the design of a sequential S-LHS estimator, both in terms of the sample allocation in each stratification, and in terms of the design of a new stratification from an existing one.

The variance of each estimator in~\eqref{eq:weighted_estimator} is not a-priori known, and needs to be estimated. 
For $\bar{N}$ sufficiently large, an (arguably conservative in view of Sect.~\ref{sec:LHS-within-SS}) approximation for the variance based on~\eqref{eq:ss_estimator:var:general} is given by
\begin{equation}
\label{eq:var_est_emp}
v_{\ell} \approx \sum_{S \in \calS_{\ell}} p_S^2 \frac{\hat{\sigma}_{S}^2}{N^{(\ell)}_S} = \sum_{S \in \calS_{\ell}} p_S^2 \frac{\hat{\sigma}_{S}^2|\calS_\ell|}{N^{(\ell)}} =
\frac{1}{\bar{N}} \sum_{S \in \calS_{\ell}} p_S^2 \hat{\sigma}_{S}^2\;,
\end{equation}
recalling that $\ell \bar{N} = N^{(\ell)}=\sum_{S\in\calS_\ell} N^{(\ell)}_S$, $N^{(\ell)}_S = N^{(\ell)}/{|\calS_\ell|}$, and $|\calS_\ell| = \ell$.
Here, $\hat{\sigma}_{S}$ is the (standard) sample standard deviation in stratum $S$, and $\calS_{\ell}$ is the stratification of estimator $\ell$ for $\ell=1,\dots, L$. 
If available, one may, of course, substitute the sample standard deviation by an LHS estimate; c.f.~\cite{Stein_87}.

For each stratum of an S-LHS member of the estimator sequence, we compute a local gPC basis and estimate the coefficients $f^{\text{gPC},S}_{\mathbf{m}}$ in Eq.~\eqref{eq:gPC_expansion} from which we evaluate the variance decomposition~\eqref{eq:Sobol-var-dec}. From the local variance decomposition and known sizes of strata, the relative effect on the estimator's variance from all subsets of dimensions can be quantified, as given by inserting Eq.~\eqref{eq:Sobol-var-dec} into Eq.~\eqref{eq:var_est_emp}: 
\begin{equation}
\label{eq:ss-var-dec}
v_{\ell} = \frac{1}{\bar{N}} \sum_{S \in \calS_{\ell}} p_S^2 \sum_{T \subset \{ 1,2,..., d\} } \sigma_{T,S}^2
\end{equation}
This decomposition, together with a rule for how current strata should be merged and/or split, will be used to design the next estimator's stratification. 
For simplicity, we henceforth assume that we bisect a single stratum along a single dimension that has the highest contribution to the estimator's variance, as given by Eq.~\eqref{eq:ss-var-dec}. For the numerical test cases considered below, this already gives a significant variance reduction between members of the sequence of estimators.

To simplify sampling, we will permit hyperrectangular (or tensor product type) strata only. The new estimator will be sampled independently without reuse of its predecessor's sample points. Hence, the sequence of estimators described above are sampled independently.
Once the sequence of hierarchically refined S-LHS estimators has been constructed, the final weighted SS-LHS-gPC estimator is assembled using the weights \eqref{KrPe_eq:optimal:weights:avg}--\eqref{KrPe_eq:optimal:weights:avg:final}. The complete estimation procedure is summarized in pseudocode in Algorithm~\ref{algo:repeated_ss}. 
\begin{algorithm}[htbp]
  \caption{Weighted SS-LHS-gPC estimators.}
  \label{algo:repeated_ss}
  \begin{algorithmic}[1]
    \State Initialize with trivial/uninformed stratification $S_1 = \setU$ (single stratum).

    \For{$\ell = 1:L$} 
         \For{$S \in \mathcal{S}_{\ell}$} 
             \State Generate an LHS of size $\bar{N}$ for $S$. 
             \State Determine local gPC basis on S.
             \State Compute local gPC coefficients from local function samples.
             \State Compute Sobol indices from gPC coefficients using Eq.~\eqref{eq:Sobol-gPC}.
         \EndFor
        
        \State Compute S-LHS estimator $\hat{\mu}_\ell := \hat{\mu}^{\text{Strat}}_{\text{LHS}, \ell}$
        
        \State Propose refined stratification $\calS_{\ell+1}$ based on Sobol variance decomposition~\eqref{eq:ss-var-dec}.
    \EndFor 
    \State Compute weights $\alpha$ using~\eqref{KrPe_eq:optimal:weights:avg}-\eqref{KrPe_eq:optimal:weights:avg:final}.
    \State Output $\sum_{\ell=1}^{L} \alpha_{\ell} \hat{\mu}_{\ell}$
  \end{algorithmic}
\end{algorithm}

The steps in Algorithm~\ref{algo:repeated_ss} have deliberately been kept as general as possible. Next, we provide information about the implementation of the steps. Note that there are typically alternative methods that can be used depending on the characteristics of the problem at hand. The LHS on stratum $S$ (L4) is generated via the method described in \cite[Ch.~9.6]{Kroese_etal_11}. For the local gPC basis (L5), first univariate basis functions are generated via the recurrence relations Eq.~\eqref{eq:rec_coeff} using Stieltje's procedure, and then multivariate basis functions are taken as all products up so some predefined total degree of polynomials according to Algorithm~\ref{algo:mult_var_basis}, where the polynomial order $p$ is set by the user.

\begin{algorithm}[htbp]
  \caption{Total-order basis from univariate gPC.}
  \label{algo:mult_var_basis}
  \begin{algorithmic}[1]
    \State $\mathcal{I}_{d,p}^{\text{tot ord}} = 
      \Bigl\{ \mathbf{m} \in \mathbb{Z}_{0+}^{n} : \left\| \mathbf{m} \right\| \leq p
    \Bigr\}$.

    \For{$\mathbf{m}  \in  \mathcal{I}_{d,p}^{\text{tot ord}}$ } 
        \State $\psi_{\mathbf{m}}^{S}(\bsY_S) = 1$
         \For{$k=1:d$} 
             \State $\psi_{\mathbf{m}}^{S}(\bsY_S) = \psi_{\mathbf{m}}^{S}(\bsY_S) \psi_{k,{m_k}}^{S}(Y_{k,S})$
         \EndFor
    \EndFor 
  \end{algorithmic}
\end{algorithm}
All relevant expressions can be found in Sect.~\ref{sec:surr_gPC}. Computation of the gPC coefficients stratum-wise (L6) is performed by solving the (over-determined) least-squares problem
\[
\argmin_{\mathbf{f}^{S}} \left\| \mathbf{A}^{S} \mathbf{f}^{S} - \mathbf{b}^{S} \right\|_{2}, \mbox{ where } [\mathbf{A}^{S}]_{j,m} = \psi_{\mathbf{m}}^{S}(\bsY_S^{(j)}), \quad \mathbf{f}^{S}_{m} = f^{\text{gPC},S}_{\mathbf{m}}, \quad 
\mathbf{b}^{S}_{j} = f(\bsY_S^{(j)}), 
\]
for samples $j=1,\dots, N_S$ and $m=1,\dots,| \mathcal{I}_{d,p}^{\text{tot ord}} |$ where $m=m(\mathbf{m})$ is any mapping from the multi-indices to single indices. The asymptotic computational complexity of solving the least-squares problem for a given stratum $S$ is $O(N_S | \mathcal{I}_{d,p}^{\text{tot ord}} |^2)$, but we emphasize that both $N_S$ and $| \mathcal{I}_{d,p}^{\text{tot ord}} |$ should be kept small.  Once the gPC coefficients are found, all Sobol indices can be directly computed by monitoring which of the $|\mathcal{I}_{d,p}^{\text{tot ord}}|$ polynomials are functions of any given subset of random variables (L7).

Finally, we reiterate that, throughout this work, we assume that evaluations of the function $f$ are computationally expensive. Therefore, the additional computational cost incurred by lines 4--7 in Algorithm~\ref{algo:repeated_ss} is considered negligible compared to the cost of acquiring samples by evaluating $f$.

\section{Numerical results}
\label{sec:numerics}

The numerical test cases below are performed with the proposed SS-LHS-gPC method, as summarized in Algorithm~\ref{algo:repeated_ss}. We use $\bar{N}=50$ samples for each stratum. The local gPC coefficients are obtained by solving a standard least squares problem, where the polynomial order is set so that the total number of basis functions is smaller than $\bar{N}$. In a situation where high-order polynomial approximation is desired, the least squares method could be replaced by, e.g., Monte Carlo integration using the same $\bar{N}$ samples independently for all gPC coefficients.  

\subsection{Test cases}
For numerical test cases, we consider the following functions: 
\begin{align}
&f(Y_1, Y_2) = \frac{1}{|a - Y_1^2-Y_2^2| +\delta}\;,
\label{eq:prob_Shields}
\tag{P1} \\
&f(Y_1,\dots, Y_d) = c \mathbbm{1}_{\bigl\{ \sum_{m=1}^{d'} Y_m^2 \leq r^2\bigr\}}\;,
\label{eq:prob_qcircle}
\tag{P2}\\
&f(Y_1,\dots, Y_d) = c \left( \mathbbm{1}_{\bigl\{ \sum_{m=1}^{d'} Y_m^2 \leq r_1^2\bigr\}} + \mathbbm{1}_{\bigl\{ \sum_{m=d'+1}^{2d'} Y_m^2 \leq r_2^2\bigr\}} \right)\;.
\label{eq:prob_2circles}
\tag{P3}
\end{align}
The case~\ref{eq:prob_Shields} with $a=0.3$ and varying $\delta$, adapted from test problems in~\cite{Shields_18}, is chosen to illustrate the effect of localization of the variance in stochastic space. The case~\ref{eq:prob_qcircle} with $r=0.4$ illustrates the effect of a $d'$-dimensional object embedded in a $d$-dimensional space. 
Case~\ref{eq:prob_2circles} with $r_1=r_2=0.4$ is effectively $d'$-dimensional in the superposition sense ($d_{\text{sup}}=d'$), and $2d'$-dimensional in the truncation sense ($d_{\text{tr}}=2d'$).

The function $f$ and sequential stratifications are shown in Figs.~\ref{fig:stratifications_P1}--\ref{fig:stratifications_P2} for some 2D test cases, namely for~\ref{eq:prob_Shields} with $\delta=1, 0.1, 0.01$, and $a=0.3$, and for~\ref{eq:prob_qcircle} with $d=d'=2$. The stratification sizes ($|\calS|=6, 20, 63$) have been chosen as they correspond to sample sizes of approximately $N=10^{3}$, $N=10^{4}$, and $N=10^{5}$, respectively. The stratifications are themselves random, and each subplot corresponds to a unique set of samples, so a finer stratification in Figs.~\ref{fig:stratifications_P1}--\ref{fig:stratifications_P2} is not necessarily a refinement of a coarser stratification for the same problem shown in this figure.
\begin{figure}[htbp]
\centering
\subfigure[
\ref{eq:prob_Shields}: $\delta=1$; $|\mathcal{S}|=6$.]
{\includegraphics[width=0.31\textwidth]{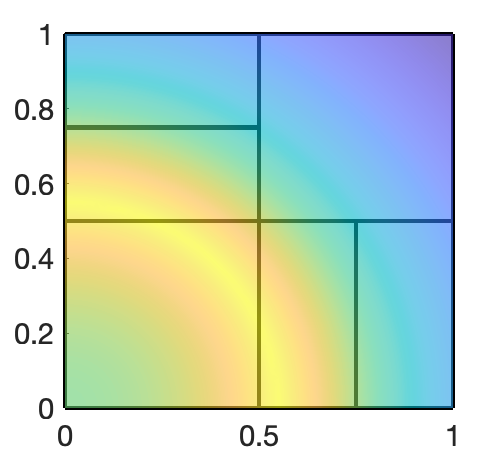}
}
~
\subfigure[
\ref{eq:prob_Shields}: $\delta=1$; $|\mathcal{S}|=20$]
{\includegraphics[width=0.31\textwidth]{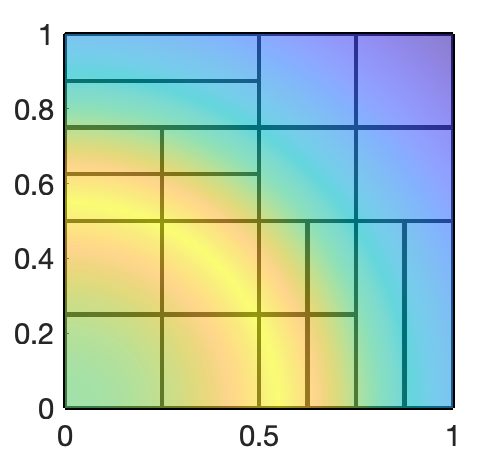}
}
~
\subfigure[
\ref{eq:prob_Shields}: $\delta=1$; $|\mathcal{S}|=63$.]
{\includegraphics[width=0.31\textwidth]{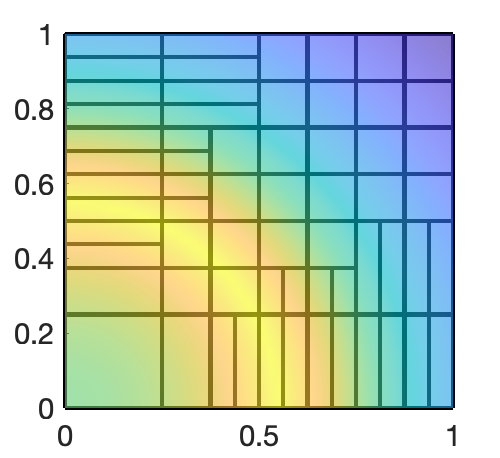}
}
\subfigure[
\ref{eq:prob_Shields}: $\delta=0.1$; $|\mathcal{S}|=6$.]
{\includegraphics[width=0.31\textwidth]{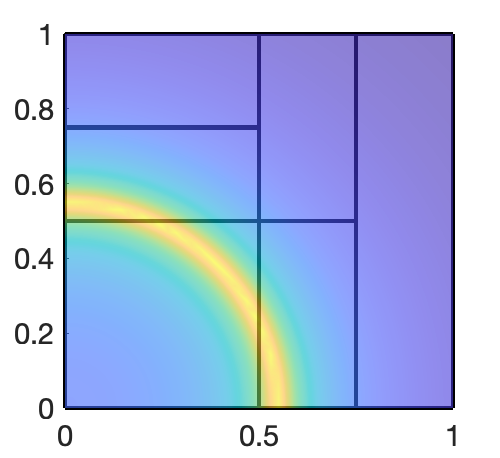}
}
~
\subfigure[
\ref{eq:prob_Shields}: $\delta=0.1$; $|\mathcal{S}|=20$]
{\includegraphics[width=0.31\textwidth]{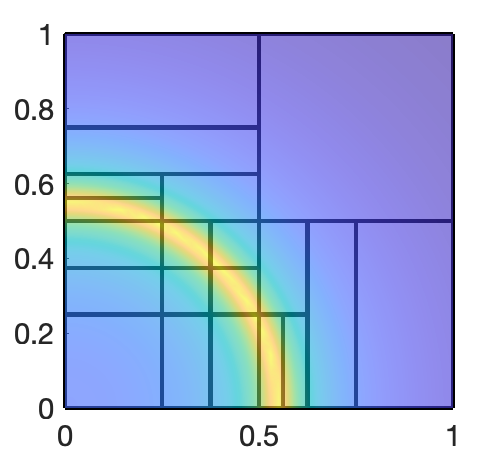}
}
~
\subfigure[
\ref{eq:prob_Shields}: $\delta=0.1$; $|\mathcal{S}|=63$.]
{\includegraphics[width=0.31\textwidth]{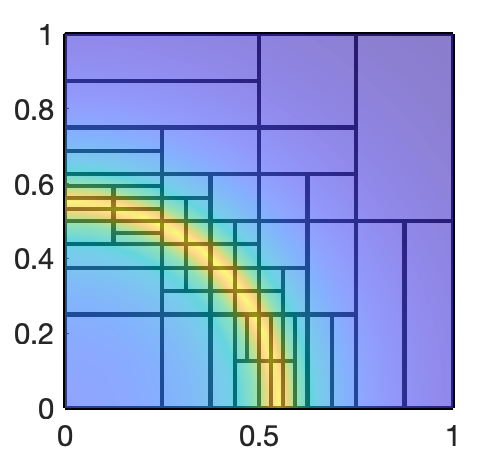}
}
\subfigure[
\ref{eq:prob_Shields}: $\delta=0.01$; $|\mathcal{S}|=6$.]
{\includegraphics[width=0.31\textwidth]{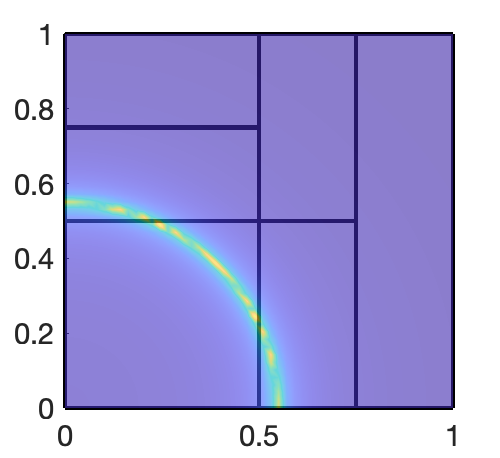}
}
~
\subfigure[
\ref{eq:prob_Shields}: $\delta=0.01$; $|\mathcal{S}|=20$]
{\includegraphics[width=0.31\textwidth]{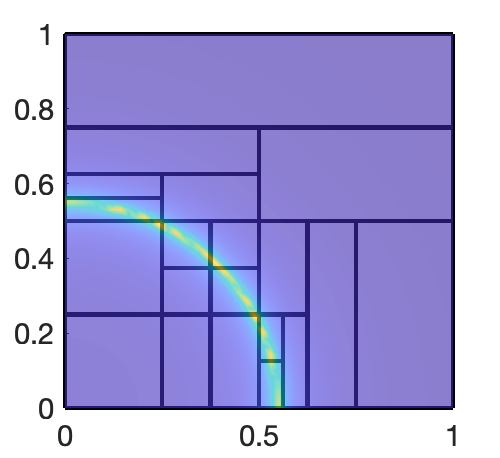}
}
~
\subfigure[
\ref{eq:prob_Shields}: $\delta=0.01$; $|\mathcal{S}|=63$.]
{\includegraphics[width=0.31\textwidth]{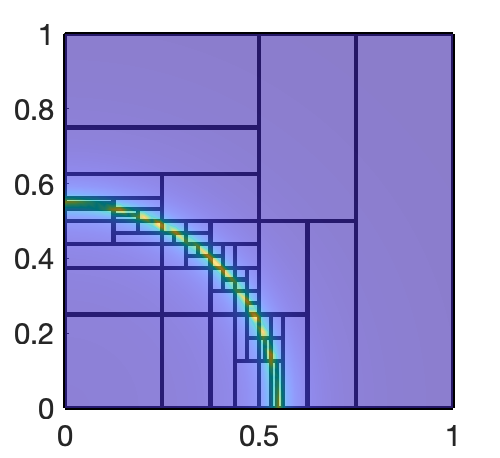}
}
\caption{Sequential stratifications for the 2D test cases: P1. }
\label{fig:stratifications_P1}
\end{figure}
\begin{figure}[htbp]
\centering
\subfigure[
\ref{eq:prob_qcircle}: $d'=d=2$; $|\mathcal{S}|=6$.]
{\includegraphics[width=0.31\textwidth]{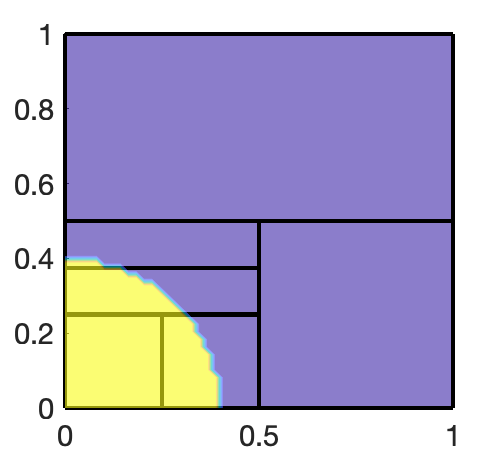}
}
~
\subfigure[
\ref{eq:prob_qcircle}: $d'=d=2$; $|\mathcal{S}|=20$]
{\includegraphics[width=0.31\textwidth]{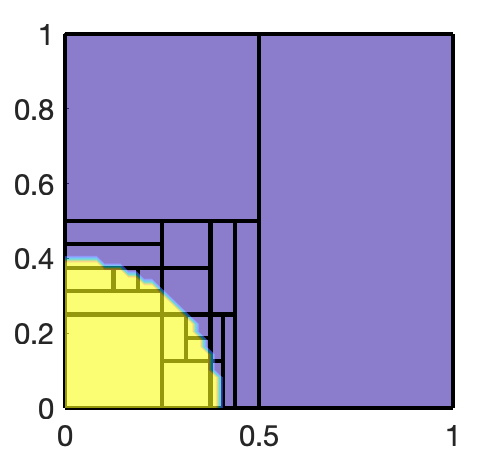}
}
~
\subfigure[
\ref{eq:prob_qcircle}: $d'=d=2$; $|\mathcal{S}|=63$]
{\includegraphics[width=0.31\textwidth]{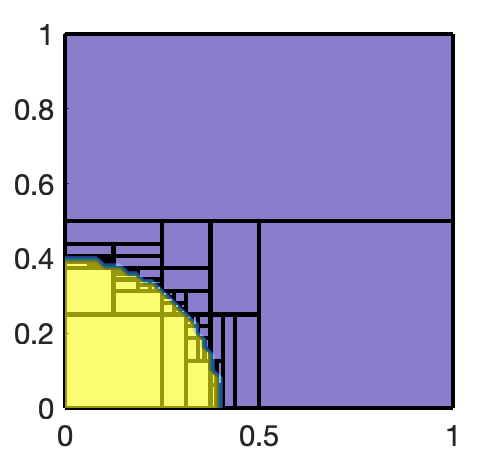}
}
\caption{Sequential stratifications for the 2D test cases: P2. }
\label{fig:stratifications_P2}
\end{figure}

For all test cases, the proposed SS-LHS-gPC is compared to standard LHS and SMC with the same number of samples. To approximate the 
variance of the estimators, each numerical experiment is evaluated independently 100 times, and the sample variances are computed. The results are shown as functions of the total number of samples in Figs.~\ref{fig:var_vs_samples_P1}--\ref{fig:var_vs_samples_P3}, together with $\calO(N^{-2})$ reference curves. For the case~\ref{eq:prob_Shields} shown in Fig.~\ref{fig:var_vs_samples_P1}, all three methods (SS-LHS-gPC, LHS, SMC) exhibit increased variance with decreasing $\delta$, and the difference between LHS and SMC decreases. SS-LHS-gPC shows an $\calO(N^{-2})$ variance decay for this test case (irrespective of $\delta$), leading to two orders of magnitude smaller variance than LHS for the largest sample sizes considered in this study.

\begin{figure}[htbp]
\centering
\subfigure[
\ref{eq:prob_Shields} with $\delta=1$.]
{\includegraphics[width=0.31\textwidth]{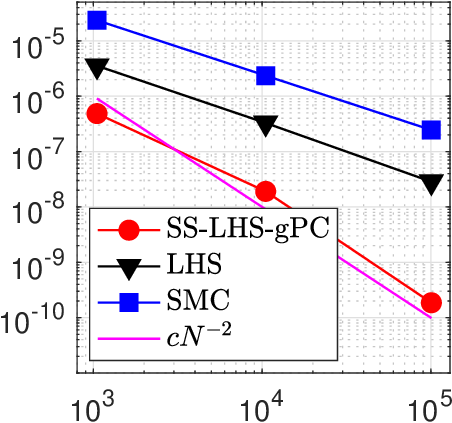}
}
~
\subfigure[
\ref{eq:prob_Shields} with $\delta=0.1$.]
{\includegraphics[width=0.31\textwidth]{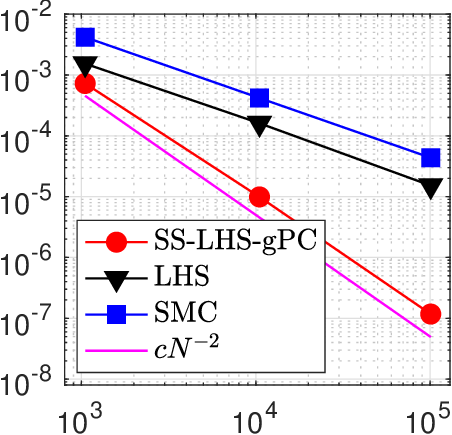}
}
~
\subfigure[
\ref{eq:prob_Shields} with $\delta=0.01$.]
{\includegraphics[width=0.31\textwidth]{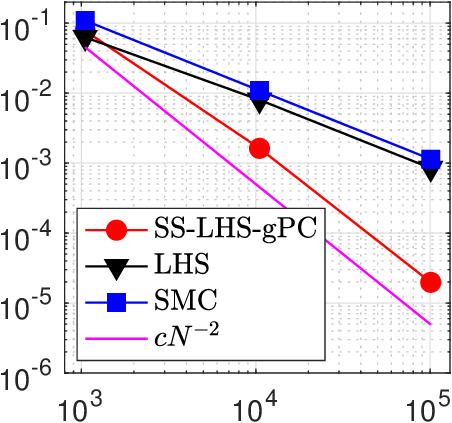}
}
\caption{P1: Variance of sequential stratifications, LHS and SMC as functions of number of samples.  
}
\label{fig:var_vs_samples_P1}
\end{figure}

The first row of Fig.~\ref{fig:var_vs_samples_P2} shows the 2D quarter-circle problem~\ref{eq:prob_qcircle} embedded in 2D, 3D, and 10D random spaces, respectively. The effective dimension (and the variance of the test function) is $d_{\text{sup}}=d_{\text{tr}}=2$, and all three methods perform almost independently of the total number of dimensions. An $\calO(N^{-2})$ variance decay for SS-LHS-gPC is observed also for this set of test cases.
The second row of Fig.~\ref{fig:var_vs_samples_P2} shows an effectively 3D version of problem~\ref{eq:prob_qcircle}, now in spaces of dimension 3, 4, and 10. LHS and SMC perform similarly to the preceding problem. Again, SS-LHS-gPC provides smaller variances, but no longer has an $\calO(N^{-2})$ variance decay, although the plots indicate that larger sample sizes possibly lead to higher variance decay rates.

\begin{figure}[htbp]
\centering
\subfigure[
\ref{eq:prob_qcircle} with $d'=2$, $d=2$.]
{\includegraphics[width=0.31\textwidth]{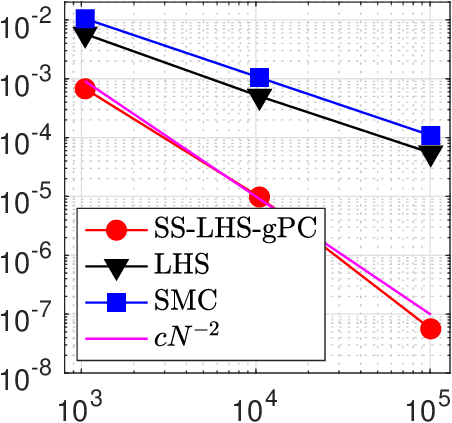}
}
~
\subfigure[
\ref{eq:prob_qcircle} with $d'=2$, $d=3$.]
{\includegraphics[width=0.31\textwidth]{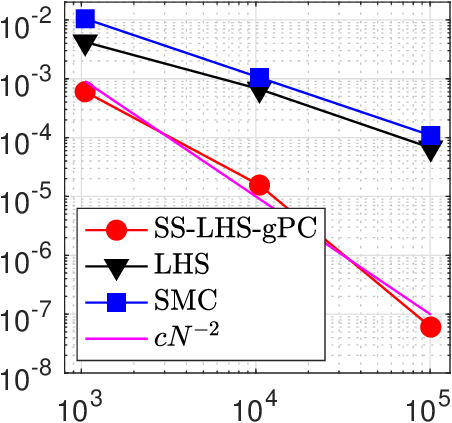}
}
~
\subfigure[
\ref{eq:prob_qcircle} with $d'=2$, $d=10$.]
{\includegraphics[width=0.31\textwidth]{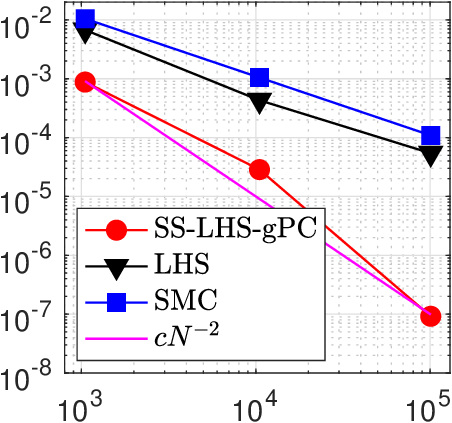}
}
~
\subfigure[
\ref{eq:prob_qcircle} with $d'=3$, $d=3$.]
{\includegraphics[width=0.31\textwidth]{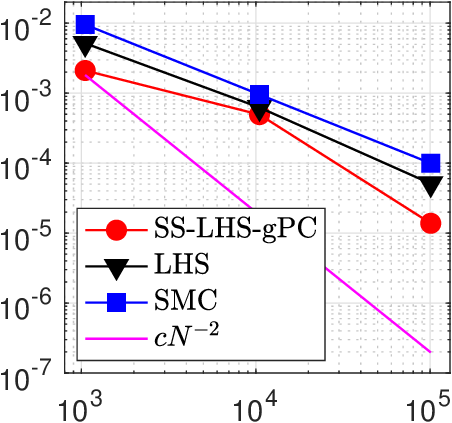}
}
~
\subfigure[
\ref{eq:prob_qcircle} with $d'=3$, $d=4$.]
{\includegraphics[width=0.31\textwidth]{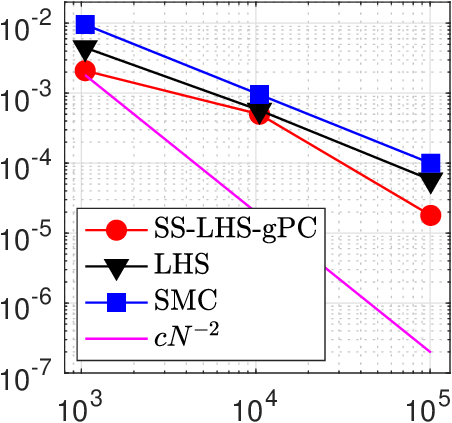}
}
~
\subfigure[
\ref{eq:prob_qcircle} with $d'=3$, $d=10$.]
{\includegraphics[width=0.31\textwidth]{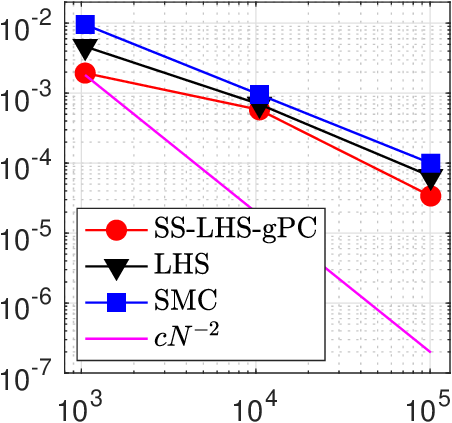}
}
\caption{P2: Variance of sequential stratifications, LHS and SMC as functions of number of samples.
}
\label{fig:var_vs_samples_P2}
\end{figure}

Figure~\ref{fig:var_vs_samples_P3} depicts problem~\ref{eq:prob_2circles} with two 2D objects in an effectively (in the truncation sense) 4D space. Again, we consider the problem to be embedded in a space of increasingly higher dimensionality. The higher variance of the functions themselves leads to slightly higher error constants for SMC and LHS. For the case with $d=10$, the performance of SS-LHS-gPC is between that of the~\ref{eq:prob_qcircle} cases with $d'=2$ and $d'=3$.  This can be attributed to the~\ref{eq:prob_2circles} case having lower effective dimension in the superposition sense than the~\ref{eq:prob_qcircle} with $d'=3$. The observed variance decay of~\ref{eq:prob_2circles} is faster than $\calO(N^{-1})$ but slower than $\calO(N^{-2})$.

\begin{figure}[htbp]
\centering
\subfigure[
\ref{eq:prob_2circles}with $d'=2$, $d=4$.]
{\includegraphics[width=0.31\textwidth]{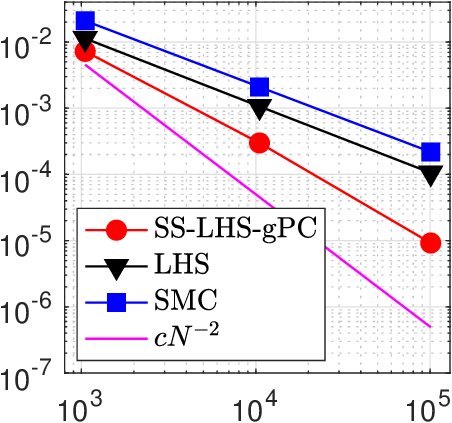}
}
~
\subfigure[
\ref{eq:prob_2circles} with $d'=2$, $d=5$.]
{\includegraphics[width=0.31\textwidth]{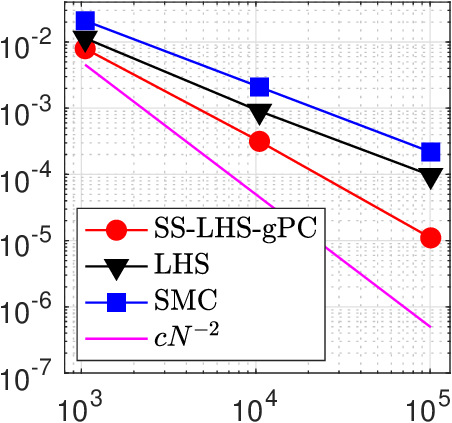}
}
~
\subfigure[
\ref{eq:prob_2circles} with $d'=2$, $d=10$.]
{\includegraphics[width=0.31\textwidth]{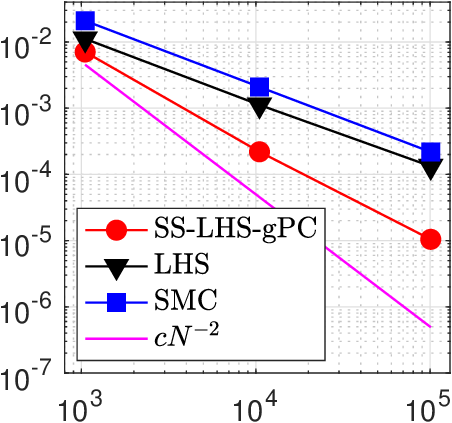}
}
\caption{P3: Variance of sequential stratifications, LHS and SMC as functions of number of samples.}
\label{fig:var_vs_samples_P3}
\end{figure}

\section{Conclusion}
\label{sec:conclusion}
We have proposed a novel sampling-based estimation method for computing statistics of random system outputs. The SS-LHS-gPC method is based on optimal weighting of a hierarchical sequence of stratified sampling estimators, where Latin hypercube sampling is employed in each stratum.
The sequence of estimators is refined through a stratum splitting criterion to reduce the maximum local variance contribution, estimated from generalized polynomial chaos expansions and Sobol sensitivity indices. Numerical results indicate that SS-LHS-gPC consistently outperforms classic Monte Carlo estimators and LHS estimators by providing estimates with smaller variance. In fact, the novel estimation procedure can even achieve a variance of
order $N^{-2}$ in some cases, where $N$ denotes the sample size. Further work is required to rigorously assess the variance reduction properties.
While this work's focus was on computational models that are significantly more expensive to evaluate than the additional post-processing incurred by Algorithm~\ref{algo:repeated_ss}, the presented methodology is also applicable in other cases. We note, however, that for applications where the computational model is cheap to evaluate, a performance comparison that goes beyond a method's variance may be required. Indeed, a performance measure that additionally accounts for the total work required, such as a method's work-normalized (relative) variance, may be more suitable.
Finally, an implementation of the method presented here is available at \url{https://github.com/massperp/Sequential-Stratified-LHS}.

\section*{Acknowledgements}
\vspace*{-0.5ex}
The authors are grateful to the anonymous reviewer for the insightful and constructive feedback that helped improve the manuscript. This work was partly funded by the Research Council of Norway through the project Expansion of Resources for CO$_2$ Storage on the Horda Platform (ExpReCCS) under project number 336294. 


\end{document}